\documentclass[runningheads]{llncs}
\usepackage{amsmath,graphicx,amssymb}

\def\email#1{{\normalsize\tt #1}}
\def\avg#1.{\langle #1\rangle}
\def\C{\mathbb C}
\def\half{\frac12}
\def\v{{\mathbf v}}

\mainmatter

\title{On Krawtchouk Transforms}
\author{Philip Feinsilver\inst{1} and Ren{\'e} Schott\inst{2}}
\institute{Southern Illinois University, 
 {Carbondale, IL. 62901, U.S.A.}
\email{pfeinsil@math.siu.edu}
\and {IECN and LORIA, 
Nancy-Universit\'e, Universit\'e Henri Poincar\'e\\
54506 Vandoeuvre-l\`es-Nancy, France}\\
\email{schott@loria.fr}
}

\begin{document} 

\maketitle

\begin{abstract}
Krawtchouk polynomials appear in a variety of contexts, most notably as orthogonal polynomials and in coding theory via the Krawtchouk transform.
We present an operator calculus formulation of the Krawtchouk transform that is suitable for computer implementation.
A positivity result for the Krawtchouk transform is shown. Then our approach is compared with the use of the Krawtchouk transform
in coding theory where it appears in MacWilliams' and Delsarte's theorems on weight enumerators. We conclude with a construction
of Krawtchouk polynomials in an arbitrary finite number of variables, orthogonal with respect to the multinomial distribution.
\end{abstract}

\thispagestyle{empty}

\section{Introduction}
Krawtchouk polynomials appear originally as orthogonal polynomials for the binomial distribution \cite{FS2,Sz}, 
and in coding theory via the Krawtchouk transform in the context of MacWilliams' theorem on weight enumerators as well in 
Delsarte's extension to association schemes \cite{H,SMcW}. 
They play a role in discrete formulations of quantum mechanics \cite{APW,Lo1,Lo2,SA},
transforms in optics \cite{AW}, as well as in recent developments in image analysis \cite{Y}. \\

We present an operator calculus formulation of the Krawtchouk transform that not only is theoretically elucidating, it is highly suitable for computer implementation.
A consequence of our formulation is a positivity theorem for the Krawtchouk transform of polynomials. We indicate connections with the
transform appearing in coding theory.\bigskip

\hrule
\noindent\flushbottom{\small This work appeared in:\\
Serge Autexier, Jacques Calmet, David Delahaye, Patrick D.~F. Ion, Laurence
  Rideau, Renaud Rioboo, and Alan~P. Sexton, editors.
\newblock {\em Intelligent Computer Mathematics, 10th International Conference,
  AISC 2010, 17th Symposium, Calculemus 2010, and 9th International Conference,
  MKM 2010, Paris, France, July 5-10, 2010. Proceedings}, volume 6167 of {\em
  Lecture Notes in Computer Science}. Springer, 2010.
}

\section{Krawtchouk polynomials as a canonical Appell system}
\subsection{Generating function}
Consider a Bernoulli random walk starting at the origin, jumping to the left with probability $q$, to the right with probability $p$, $p+q=1$, $pq\ne0$.
After $N$ steps, the position is $x$, with $j=(N-x)/2$ denoting the number of jumps to the left.\bigskip

Start with the generating function
\begin{align*}
G(v)&=(1+2qv)^{(N+x)/2}(1-2pv)^{(N-x)/2}\\ &=(1+2qv)^{N-j}(1-2pv)^j=\sum_{n=0}^N \frac{v^n}{n!}K_n
\end{align*}

where we consider $K_n$ as a function of $x$ or of $j$ according to context. \bigskip

\subsection{Orthogonality}\label{sec:ortho}
We check orthogonality with respect to the binomial distribution, $j$ running from $0$ to $N$, with corresponding probabilities 
$\displaystyle \binom{N}{j}q^jp^{N-j}$. Using angle brackets to denote expected value, we wish to show that
$\avg G(v)G(w).$ is a function of the product $vw$. 

\begin{align*}
\avg G(v)G(w).&=\sum\binom{N}{j}q^jp^{N-j}(1+2qv)^{N-j}(1-2pv)^j(1+2qw)^{N-j}(1-2pw)^j\cr
&=\sum\binom{N}{j}q^jp^{N-j}(1+2q(v+w)+4q^2vw)^{N-j}(1-2p(v+w)+4p^2vw)^{j}\cr
&=(p+2pq(v+w)+4pq^2vw+q-2pq(v+w)+4p^2qvw)^N\cr &=(1+4pqvw)^N
\end{align*}
which immediately gives the squared norms $\|K_n\|^2=(n!)^2\binom{N}{n}(4pq)^n$.
\subsection{Canonical Appell system}
For an Appell system with generating function $\exp[zx-tH(z)]$, a corresponding canonical Appell system has a generating function of the form
$$\exp[xU(v)-tH(U(v))]$$
where $U(v)$ is analytic about the origin in $\C$, with analytic inverse $V(z)$, and $H(z)$ is the logarithm of the Fourier-Laplace transform
of the distribution of $x$ at time $1$. Here we have $N$ replacing $t$, and write
$$G(v)=(1+2(q-p)v-4pqv^2)^{N/2}\,\left(\frac{1+2qv}{1-2pv}\right)^{x/2}$$
identifying
$$U(v)=\half\,\log\frac{1+2qv}{1-2pv}\,,\quad \text{and}\quad H(z)=\log(pe^z+qe^{-z})$$
One checks that
$$\log(pe^{U(v)}+qe^{-U(v)})=p\sqrt{\frac{1+2qv}{1-2pv}}+q\sqrt{\frac{1-2pv}{1+2qv}} =(1+2(q-p)v-4pqv^2)^{-1/2}$$
which verifies the form $\exp[xU(v)-NH(U(v))]$.\bigskip

Solving $z=U(v)=U(V(z))$, we find
$$V(z)=\frac{(e^2-e^{-z})/2)}{pe^z+qe^{-z}}=\frac{\sinh z}{pe^z+qe^{-z}}$$
\section{Krawtchouk expansions}
The generating function, with $z=U(v)$, is
$$G(v)=e^{zx-NH(z)}=\sum_{n\ge0}\frac{V(z)^n}{n!}\,K_n(x,N)$$
Rearrange to get,
$$e^{zx}=(pe^z+qe^{-z})^N\sum_{n\ge0}\left(\frac{\sinh z}{pe^z+qe^{-z}}\right)^n\,\frac{K_n(x,N)}{n!}$$

We want to write the coefficients of the expansion in terms of $D=d/dx$ acting on a function of $x$. We cannot
substitute $z\leftrightarrow D$ directly, since the $D$ and $x$ do not commute. Introduce another variable $s$.
Replacing $z$ by $D_s=d/ds$, apply both sides to a function $f(s)$:
$$e^{xD_s}f(s)=f(s+x)=(pe^{D_s}+qe^{-D_s})^N\sum_{n\ge0}\left(\frac{\sinh D_s}{pe^{D_s}+qe^{-D_s}}\right)^n\,\frac{K_n(x,N)}{n!}f(s)$$
Now we move the operators involving $D_s$ past $K_n(x,N)$.
Letting $s=0$, thinking of $f$ as a function of $x$ instead of $s$, 
we can replace $D_s$ by our usual $D=d/dx$,  to get
$$f(x)=\sum_{0\le n\le N}\frac{K_n(x,N)}{n!}(pe^{D}+qe^{-D})^N\left(\frac{\sinh D}{pe^{D}+qe^{-D}}\right)^nf(0)$$
In other words, the coefficients of the Krawtchouk expansion of $f(x)$ are given by
\begin{equation}\label{eq:kexp}
\tilde f(n)=\frac{1}{n!}\,(pe^{D}+qe^{-D})^{N-n}(\sinh D)^nf(0)
\end{equation}
\begin{theorem} For $p>q$, if a polynomial has positive coefficients, then the coefficients of its Krawtchouk expansion are positive. 
For $p=q=1/2$, we have nonnegativity of the Krawtchouk coefficients.
\end{theorem}
\begin{proof}
Note that if $\psi(D)$ is any formal power series $\sum c_nD^n$ in $D$, then
$$\psi(D)x^m\bigm|_0=m!\,c_m$$
is positive if $c_m$ is. Now, write $pe^z+qe^{-z}=\cosh z+(p-q)\sinh z$. If $p>q$, then the power series expansion has positive coefficients so that 
the coefficients in the Krawtchouk expansion are positive as well. For $p=q$, we have nonnegativity.
\end{proof}

\subsection{Matrix formulation}
As shown in \cite{FS1}, for Krawtchouk polynomials of degree at most $N$, one can use the matrix of $D$, denoted $\hat D$, acting on
the standard basis $\{1, x, x^2,\ldots,x^N\}$ replacing the operator $D$ in eq.\,\eqref{eq:kexp}.\\

For example, take $N=4$. We have
$$\hat D= \begin{pmatrix}0&1&0&0&0 \cr 0&0&2&0&0 \cr 0&0&0&3&0\cr 0&0&0&0&4\cr 0&0&0&0&0 \end{pmatrix}$$
The basic matrices needed are computed directly from the power series for the exponential function. Thus,
$$pe^{\hat D}+qe^{-\hat D}=
 \left( \begin {array}{ccccc} 1&2\,p-1&1&2\,p-1&1\\\noalign{\medskip}0
&1&4\,p-2&3&8\,p-4\\\noalign{\medskip}0&0&1&6\,p-3&6
\\\noalign{\medskip}0&0&0&1&8\,p-4\\\noalign{\medskip}0&0&0&0&1
\end {array} \right)$$
where the relation $p+q=1$ has been used and
$$\sinh\hat D= \left( \begin {array}{ccccc} 0&1&0&1&0\\\noalign{\medskip}0&0&2&0&4
\\\noalign{\medskip}0&0&0&3&0\\\noalign{\medskip}0&0&0&0&4
\\\noalign{\medskip}0&0&0&0&0\end {array} \right)$$

The nilpotence of $\hat D$ reduces the exponentials to polynomials in $\hat D$, making computations with these matrices convenient and fast.

\subsection{Functions of two or more variables}
For functions on the plane, on an $(N+1)$-by-$(M+1)$ grid, we can use the product basis $K_n(x,N)\,K_m(y,M)$. The formula for the expansion coefficients
has the product form, denoting $D_1=\partial/\partial x$, $D_2=\partial/\partial y$,
$$\tilde f(n,m)=\frac{1}{n!\,m!}\,(p_1e^{D_1}+q_1e^{-D_1})^{N-n}(p_2e^{D_2}+q_2e^{-D_2})^{M-m}(\sinh D_1)^n(\sinh D_2)^mf(0,0)$$
where we note the relations $p_1+q_1=p_2+q_2=1$. The basis functions are orthogonal with respect to the corresponding
product of the marginal binomial distributions, i.e. the underlying random variables are independent.\bigskip

A similar construction works for 3 or more variables. In \S5, we present an extension to multinomial distributions.

\section{Krawtchouk polynomials in coding theory}

The Krawtchouk transform appears in coding theory in a different variant as an essential component of MacWilliams' theorem
on weight enumerators \cite{H,SMcW}. It appears in Delsarte's formulation in terms of association schemes as well. 

Fix $N>0$. For Krawtchouk transforms on functions defined on $\{-N, 2-N, \ldots, N-2, N\}$ we use the Fourier-Krawtchouk matrices, $\Phi$,
which we call ``Kravchuk matrices". The entries are the values of the Krawtchouk polynomials as functions on $\{-N,\ldots,N\}$. Thus,
via the mapping $x=N-2j$, $0\le j\le N$, the columns correspond to values of $x$ and we write
$$G(v)=(1+2qv)^{N-j}(1-2pv)^j=\sum_i v^i\Phi_{ij}$$

In \cite{H}, the Krawtchouk polynomials are defined via a slightly different generating function, changing the notation to fit our context,
$$G_s(v)=(1+(s-1)v)^{N-j}(1-v)^j=\sum_i v^i K_i(j;N,s)$$
The difference is primarily one of scaling. Comparing $G$ with $G_s$, replacing $v$ by $v/(2p)$, we find
$$K_i(j;N,s)=(2p)^{-i}\Phi_{ij}$$
with $s-1=\frac{q}{p}=\frac{1-p}{p}=\frac{1}{p}-1$, or
 $$s=\frac{1}{p}$$
The condition $s\ge 2$ thus corresponds to $p\le q$, complementary to the condition required for positivity in the previous section.\\

Following \cite[p.\,132]{H}, we have:\bigskip
\begin{itemize}     
\item MacWilliams' Theorem: If $A$ is a linear code over $\mathbb{F}_s$ and $B=A^\bot$, its dual, then the weight distribution of $B$ is, up to a factor,
the Krawtchouk transform of the weight distribution of $A$.

\item Delsarte's Theorem: If $A$ is a code over an alphabet of size $s$, then the values of the Krawtchouk transform of
the coefficients of the distance enumerator are nonnegative.
\end{itemize}

In this context, the components of the Krawtchouk transform of a vector $\v$ are defined by
$$\hat v_i=\sum_j K_i(j;N,s)v_j$$

We show how to invert the transform.

\subsection{Inverse transform}

The calculation of \S\ref{sec:ortho} can be recast in matrix form. Set $B$ equal to the diagonal matrix
$$B=\text{diag}\,(p^N,Np^{N-1}q,\ldots,\binom{N}{i}\,p^{N-i}q^i,\ldots,q^N)$$
Let $\Gamma$ denote the diagonal matrix of squared norms, here without the factors of $n!$,
$$\Gamma=\text{diag}\,(1,N(4pq),\ldots,\binom{N}{i}\,(4pq)^i,\ldots,(4pq)^N)$$

With $G(v)=\sum v^i\Phi_{ij}$, we have, following the calculation of \S\ref{sec:ortho},
\begin{align*}
\sum_{i,j} v^iw^j(\Phi B\Phi^T)_{ij}&=\sum_{i,j,k}v^i w^j\Phi_{ik}B_{kk}\Phi_{jk}\\
&=\avg G(v)G(w).=\sum_i(vw)^i\Gamma_{ii}
\end{align*}

In other words, the orthogonality relation takes the form

$$\Phi B\Phi^T=\Gamma$$
from which the inverse is immediate
$$\Phi^{-1}=B\Phi^T\Gamma^{-1}$$

A consequence of this formula is that $\det \Phi$ is independent of $p$ and $q$:

$$(\det\Phi)^2=\frac{\det \Gamma}{\det B}=2^{N(N+1)}$$
The sign can be checked for the symmetric case $p=q=1/2$ with the result
$$\det \Phi=\pm\,2^{N(N+1)/2}$$
with the $+$ sign for $N\equiv 0,3 \pmod 4$.\\

We illustrate with an example.\\

Example. For $N=4$,
$$ \Phi= \left( \begin {array}{ccccc} 1&1&1&1&1\\\noalign{\medskip}8\,q&6\,q-2
\,p&4\,q-4\,p&2\,q-6\,p&-8\,p\\\noalign{\medskip}24\,{q}^{2}&12\,{q}^{
2}-12\,pq&4\,{q}^{2}-16\,pq+4\,{p}^{2}&-12\,pq+12\,{p}^{2}&24\,{p}^{2}
\\\noalign{\medskip}32\,{q}^{3}&8\,{q}^{3}-24\,p{q}^{2}&-16\,p{q}^{2}+
16\,{p}^{2}q&24\,{p}^{2}q-8\,{p}^{3}&-32\,{p}^{3}\\\noalign{\medskip}
16\,{q}^{4}&-16\,p{q}^{3}&16\,{p}^{2}{q}^{2}&-16\,{p}^{3}q&16\,{p}^{4}
\end {array} \right)$$
where we keep $p$ and $q$ to show the symmetry.

$$Q=2^4\,B\Phi^T\Gamma^{-1}=\left( \begin {array}{ccccc} 16\,{p}^{4}&8\,{p}^{3}&4\,{p}^{2}&2\,p&1
\\\noalign{\medskip}64\,{p}^{3}q&8\,{p}^{2} \left( 3\,q-p \right) &8\,
p \left( q-p \right) &2\,q-6\,p&-4\\\noalign{\medskip}96\,{p}^{2}{q}^{
2}&24\,qp \left( q-p \right) &4\,{q}^{2}-16\,pq+4\,{p}^{2}&-6\,q+6\,p&
6\\\noalign{\medskip}64\,p{q}^{3}&8\,{q}^{2} \left( q-3\,p \right) &-8
\,q \left( q-p \right) &6\,q-2\,p&-4\\\noalign{\medskip}16\,{q}^{4}&-8
\,{q}^{3}&4\,{q}^{2}&-2\,q&1\end {array} \right)$$

satisfies $Q\Phi=\Phi Q=2^4\,I$.

\subsection{Modified involution property}
For $p=1/2$, there is a simpler approach to inversion. Namely the identity
$$\Phi^2=2^N\,I$$
We will see how this is modified for $p\ne1/2$, which will result in a very simple form for $\Phi^{-1}$.

\begin{proposition}
Let $P$ be the diagonal matrix 
$$P=\mbox{\rm diag}\,((2p)^N,\ldots,(2p)^{N-j},\ldots,1)$$
Let $P'$ be the diagonal matrix
$$P'=\text{\rm diag}\,(1,\ldots,(2p)^j,\ldots,(2p)^N)$$
Then
$$\Phi P\Phi=2^NP'$$
\end{proposition}

We just sketch the proof as it is similar to that for orthogonality.

\begin{proof}
The matrix equation is the same as the corresponding identity via generating functions. Namely,
$$\sum_{i,j,k}v^i\Phi_{ik}(2p)^{N-k}\Phi_{kj}w^j\binom{N}{j}=2^N(1+2pvw)^N$$
First, sum over $i$, using the generating function $G(v)$, with $j$ replaced by $k$. Then sum over $k$, again using the generating 
function. Finally, summing over $j$ using the binomial theorem yields the desired result, via $p+q=1$.
\end{proof}

Thus,

\begin{corollary}
$$\Phi^{-1}=2^{-N}\,P\Phi P'^{-1}$$
\end{corollary}

\begin{section}{Krawtchouk polynomials in 2 or more variables}
Now we will show a general construction of Krawtchouk polynomials in variables $(j_1,j_2,\ldots,j_d)$, running
from $0$ to $N$ according to the level $N$. These systems are analogous to wavelets in that they have a dimension, $d$,
and a resolution $N$.
\subsection{Symmetric representation of a matrix}
Given a $d\times d$ matrix $A$, we will find the ``symmetric representation" of $A$, the action on
the symmetric tensor algebra of the underlying vector space. This is effectively the action 
of the matrix $A$ on vectors extended to polynomials in $d$ variables. 

Introduce commuting variables $x_1,\ldots,x_d$. Map
$$y_i=\sum_{j}A_{ij}x_j$$
We use multi-indices,
$m=m_1,\ldots,m_d$, $m_i\ge0$, similarly for $n$. Then a monomial
$$x^m=x_1^{m_1}x_2^{m_2}\cdots x_d^{m_d}$$
and similarly for $y^n$. The induced map at level $N$ has matrix elements $\bar A_{nm}$ determined by the expansion
$$ y^n=\sum_{m} \bar A_{nm}x^m$$
One can order the matrix entries lexicographically corresponding to the monomials $x^m$.
We call this the \textsl{induced matrix} at level $N$. Often the level is called the \textsl{degree} since
the induced matrix maps monomials of homogeneous degree $N$ to polynomials of homogeneous degree $N$.

We introduce the special matrix $B$ which is a diagonal matrix with multinomial coefficients as entries.
$$B_{nm}=\delta_{nm}\binom{N}{n}=\frac{N!}{n_1!\,n_2!\,\cdots n_d!}$$
For simplicity, for $B$ we will not explicitly denote the level or dimension, as it must be consistent with the context.

The main feature of the map $A\to \bar A$ is that at each level it is a multiplicative homomorphism, that is,
$$\overline{A_1A_2}=\bar A_1\,\bar A_2$$
as follows by applying the definition to $y_i=\sum (A_1A_2)_{ij}x_j$ first to $A_1$, then to $A_2$. Observe, then, that the identity map is preserved
and that inverses map to inverses. However, transposes require a separate treatment.

\subsubsection{Transposed matrix. }
The basic lemma is the relation between the induced matrix of $A$ with that of its transpose.
We denote the transpose of $A$, e.\,g.\,, by $A^T$. 

\begin{lemma}
The induced matrices at each level satisfy
$$ \overline{A^T}=B^{-1}{\bar A}^T B$$
\end{lemma}
\begin{proof}
Start with the bilinear form $\displaystyle F=\sum_{i,j}x_iA_{ij}y_j$. Then, from the definition of $\bar A$,
$$F^N=\sum_{n,m} x^n \bar A_{nm} y^m \binom{N}{n}$$
Now write $\displaystyle F=\sum_{i,j}x_i(A^T)_{ji}\,y_j$ with
$$F^N=\sum_{n,m} \binom{N}{m}\,y^m\overline{A^T}_{mn}\,x^n$$

Matching the above expressions, we have, switching indices appropriately,
$$({\bar A}^T B)_{mn}=\bar A_{nm}\binom{N}{n}=\binom{N}{m}\, \overline{A^T}_{mn}=(B\, \overline{A^T})_{mn}$$
which is the required relation.
\end{proof}

\subsection{General construction of orthogonal polynomials with respect to a multinomial distribution}
For the remainder of the article, we work in $d+1$ dimensions, with the first coordinate subscripted with $0$.

The multinomial distribution extends the binomial distribution to a sequence of independent random variables where
one of $d$ choices occurs at each step, choice $i$ occurring with probability $p_i$. The probability of none of the $d$ choices
is $p_0=1-p_1-p_2-\cdots-p_d$. The probabilities for a multinomial distribution at step $N$are given by
$$p(j_1,j_2,\ldots,j_d)= \binom{N}{j_0,j_1,j_2,\ldots,j_d} p_0^{j_0}p_1^{j_1}\cdots p_d^{j_d}$$
this being the joint probability distribution that after $N$ trials, choice $i$ has occurred $j_i$ times, with none of them occurring
$j_0=N-j_1-\cdots-j_d$ times. Let $P$ denote the diagonal matrix with diagonal $(p_0,p_1,\ldots,p_d)$. Then with
$$y_i=\sum_j P_{ij} x_j=p_ix_i$$
we see that $y^n=p^nx^n$ which implies that, at each level $N$, $\bar P$ is diagonal with entries
$$\bar P_{nm}=\delta_{nm}p^n=p_0^{N-n_1-\cdots-n_d}p_1^{n_1}\cdots p_d^{n_d}$$
In other words, the multinomial distribution comprises the entries of the diagonal matrix $B\bar P$.

Now for the construction. Start with an orthogonal matrix $W$, a diagonal matrix of probabilities $P$, and 
a diagonal matrix $D$ with positive entries on the diagonal. Let
\begin{equation}\label{eq:start}
A=P^{-1/2}\,W\,D^{1/2}
\end{equation}
where the power on a diagonal matrix is applied entrywise. Then it is readily checked that
$$A^TPA=D$$
which will generate the squared norms of the polynomials, to be seen shortly. Using the homomorphism property, we have, via the Lemma,
$$ \overline{A^T} \bar P \bar A=B^{-1}{\bar A}^TB\bar P \bar A=\bar D$$
Or, setting $\Phi={\bar A}^T$,
\begin{equation}\label{eq:krav}
\Phi\, B\bar P\, \Phi^T=B\bar D
\end{equation}
with both $B\bar P$ and $B\bar D$ diagonal.

Taking as variables the column indices, writing $j$ for $m$, we have the $n^{\text{th}}$ Krawtchouk polynomial
$$K_n(j;N,p)=\Phi_{nj}$$
at level $N$. The relation given by equation \eqref{eq:krav} is the statement that \bigskip

\textit{
the Krawtchouk polynomials are orthogonal with respect
to the multinomial distribution, $B\bar P$, with squared norms given by the entries of the induced matrix $B\bar D$.
}\bigskip

To summarize, starting with the matrix $A$ as in equation \eqref{eq:start}, form the induced matrix at level $N$. Then the
corresponding Krawtchouk polynomials are functions of the column labels of the transpose of the induced matrix.
Apart from the labelling conventions, then, in fact they are polynomials in the row labels of the original induced matrix.

Finally, note that the basic case arises from the choice
$$A=\begin{pmatrix} 1&1\\ 1& -1\end{pmatrix}$$
for $d=1$ with $p=1-p=1/2$, $D=I$.

\begin{remark} A useful way to get a symmetric orthogonal matrix $W$ is to start with any vector, $v$, form the rank-one projection,
$V=vv^T/v^Tv$ and take for $W$ the corresponding reflection $2V-I$.
\end{remark}

\subsection{Examples}
\subsubsection{Two variables. }

For two variables, start with the $3\times3$ matrices

$$A= \left( \begin {array}{rrr} 1&1&1\\\noalign{\medskip}1&-1&0
\\\noalign{\medskip}1&1&-2\end {array} \right) \,,\qquad
P= \left( \begin {array}{ccc} 1/3&0&0\\\noalign{\medskip}0&1/2&0
\\\noalign{\medskip}0&0&1/6\end {array} \right) $$
and $D$ the identity. We find for the level two induced matrix
$$\Phi^{(2)}=\left( \begin {array}{rrrrrr} 1&1&1&1&1&1\\\noalign{\medskip}2&0&2&-2
&0&2\\\noalign{\medskip}2&1&-1&0&-2&-4\\\noalign{\medskip}1&-1&1&1&-1&
1\\\noalign{\medskip}2&-1&-1&0&2&-4\\\noalign{\medskip}1&0&-2&0&0&4
\end {array} \right)$$
indicating the level explicitly. These are the values of the polynomials evaluated at integer values of the variables $(j_1,j_2,\ldots)$.
So multiplying a data vector on either side will give the corresponding Krawtchouk transform of that vector.
\subsubsection{Three variables. }
This example is very close to the basic case for $d=1$. Start with the vector $v^T=(1,-1,-1,-1)$. Form the corresponding
rank-one projection and the associated reflection, as indicated in the remark above. We find
$$A= \left( \begin {array}{rrrr} 1&1&1&1\\\noalign{\medskip}1&1&-1&-1
\\\noalign{\medskip}1&-1&1&-1\\\noalign{\medskip}1&-1&-1&1\end {array} \right)$$

and take the uniform distribution $p_i=1/4$, and $D=I$. We find
$$\Phi^{(2)}= \left( \begin {array}{rrrrrrrrrr} 1&1&1&1&1&1&1&1&1&1
\\\noalign{\medskip}2&2&0&0&2&0&0&-2&-2&-2\\\noalign{\medskip}2&0&2&0&
-2&0&-2&2&0&-2\\\noalign{\medskip}2&0&0&2&-2&-2&0&-2&0&2
\\\noalign{\medskip}1&1&-1&-1&1&-1&-1&1&1&1\\\noalign{\medskip}2&0&0&-
2&-2&2&0&-2&0&2\\\noalign{\medskip}2&0&-2&0&-2&0&2&2&0&-2
\\\noalign{\medskip}1&-1&1&-1&1&-1&1&1&-1&1\\\noalign{\medskip}2&-2&0&0
&2&0&0&-2&2&-2\\\noalign{\medskip}1&-1&-1&1&1&1&-1&1&-1&1\end {array} \right)$$
With $A^2=4I$, we have, in addition to the orthogonality relation, as in equation \eqref{eq:krav}, that $(\Phi^{(2)})^2=16I$.
\end{section}

\begin{section}{Conclusion}
Using matrix methods allows for a clear formulation of the properties of Krawtchouk polynomials and
Krawtchouk transforms. Working with polynomials or with vectors, computations can be done very
efficiently. We have shown how to construct Krawtchouk polynomials in an arbitrary (finite) number of variables,
with enough flexibility in the parameters to allow for a wide range of potential applications.
\end{section}

\section{Appendix}
Here is maple code for producing the symmetric powers of a matrix.
The arguments are the matrix $X$ and the level $N$, denoted dg in the code, for ``degree".

\begin{verbatim}
  SYMPOWER := proc(X, dg)
  local nd, ND, XX, x, y, strt, i, vv, yy, j, ww,kx,kk;
    nd := (linalg:-rowdim)(X);
    ND := (combinat:-binomial)(nd + dg - 1, dg);
    XX := matrix(ND, ND, 0);
    x := vector(nd);
    y := (linalg:-multiply)(X, x);
    strt := (combinat:-binomial)(nd + dg - 1, dg - 1) - 1;
    for i to ND do vv := (combinat:-inttovec)(strt + i, nd);
      yy := product(y[kk]^vv[kk], kk = 1 .. nd);
      for j to ND do ww := (combinat:-inttovec)(strt + j, nd);
        XX[i, j] := coeftayl(yy, 
        [seq(x[kx], kx = 1 .. nd)] = [seq(0, kx = 1 .. nd)], ww);
      end do;
    end do;
    evalm(XX);
  end:"outputs the symmetric power of X";
\end{verbatim}

\end{document}